\documentclass[final, 11pt]{article}

\usepackage{ltexpprt} \usepackage{amsmath} \usepackage{amssymb} \usepackage{colordvi} \usepackage{color} \usepackage{epsf} \usepackage{pst-grad} \usepackage{pifont} \usepackage{mathcomp} \usepackage{epsfig} \usepackage{fancyhdr} \usepackage{graphics} \usepackage{wasysym} \usepackage{algorithm} \usepackage{algorithmic} \usepackage{multirow}

    \newtheorem{definition}{Definition}  \newtheorem{assumption}{Assumption}

\begin{document}

\pagestyle{plain}

\title{Online Maximizing Weighted Throughput In A Fading Channel\thanks{Research is partially supported by Seed Award from the Office of the Vice President for Research and Economic Development at George Mason University.}}

\author{Fei Li \ \ \ Zhi Zhang \\ Department of Computer Science \\ George Mason University \\ Email: {\tt \{lifei, zzhang8\}@cs.gmu.edu.}}

\maketitle


\begin{abstract}
We consider online scheduling weighted packets with time constraints over a fading channel. Packets arrive at the transmitter in an online manner. Each packet has a value and a deadline by which it should be sent. The fade state of the channel determines the throughput obtained per unit of time and the channel's quality may change over time. In this paper,  we design online algorithms to maximize weighted throughput, defined as the total value of the packets sent by their respective deadlines. Competitive ratio is employed to measure an online algorithm's performance. For this problem and one of its variants, we present two online algorithms with competitive ratios $2.618$ and $2$ respectively.
\end{abstract}


\section{Introduction}

Time-varying signal strength is a fundamental characteristic of wireless channels. Scheduling packets over fading wireless channels has received much attention (see~\cite{TH98, FMT06, TSZM07, ZM08, CNM08} and the references therein). A scheduling algorithm can significantly improve the communication performance by taking advantages of the changing channel states. In this paper, we consider scheduling weighted packets with time constraints in an online manner.

Resource allocation for fading channels has been a well-studied topic in the area of information theory. The quantity to maximize is often the Shannon capacity, defined as the tightest upper bound of the amount of information (the total number of packets) that can be reliably transmitted over a communication channel. Tse and Hanly~\cite{TH98} have found capacity limits and optimal resource allocation policies for such fading channels. In~\cite{FMT06}, the authors applied a dynamic programming approach to get the optimal solution for scheduling uniform-value packets under both time and energy constraints. A polynomial-time optimal offline solution of scheduling packets with deadlines was given in~\cite{TSZM07, ZM08}. In their problem settings, energy is minimized under the assumption that all arriving packets are successfully delivered. An optimal offline algorithm maximizing throughput and a heuristic online approach of scheduling uniform-value packets with possibly different deadlines were given in~\cite{CNM08}. No theoretical analysis has been provided for the heuristic online solution. The first work considering scheduling weighted packets is~\cite{LZ09a}, in which an optimal offline algorithm is provided but theoretical analysis of the proposed online algorithm is missing.

Note that in previous studies, packets have uniform values (except in~\cite{LZ09a}) and their arrivals at the transmitter are usually modeled by a Poisson distribution. However, packets from different users and various applications may have different significance levels of embedded information. For the sake of being realistic and practical, we associate packets with {\em weights} ({\em values}) to indicate the significance of their embedded information. We also associate packets with deadlines to represent the information's time sensitivity. None of the previous algorithms for delivering packets can be generalized to this problem setting, because a schedule with the maximum throughput does not imply its optimality on maximizing weighted throughput. In this paper, we design competitive online algorithms to maximize weighted throughput for packets with time constraints over a fading channel.


\section{Model}

We consider scheduling weighted packets with deadlines over a wireless fading channel in an online manner. In this model, time is assumed to be discrete. Each unit of time is called a {\em time step} and a few continuous time steps are called a {\em time interval}. Packets are released over time. All packets are with the same length $l \in \mathbb R^+$ ($l$ is a constant). Each packet $p$ has an integer {\em release time} ({\em arriving time}) $r_p \in \mathbb Z^+$, a positive real value $w_p \in \mathbb R^+$ to represent its {\em weight} ({\em value}), and an integer deadline $d_p \in \mathbb Z^+$ to denote the time by which it should be delivered. The time required to send a packet depends on the {\em state quality} $q_t$ ($q_t \in [q_{\min}, \ q_{\max}]$) of the fading channel during a time step $t$. For simplicity, we assume $q_{\min} = 0$ and $l = q_{\max}$ (if the fading channel is at its highest quality $q_{\max}$, one packet can be sent in a time step). Without loss of generality, we assume the fade state in a single time step keeps unchanged.  A packet has to be sent in consecutive time steps. Sending a packet $p$ takes $t(p)$ time steps subject to
\begin{displaymath}
\sum^{t_2}_{t = t_1} q_t \ge l, t_1, t_2 \in \mathbb Z^+,
\end{displaymath}
where $t(p) = t_2 - t_1$. Two or more packets cannot share (i.e., to be sent in) the same time step. If a packet $p$ is sent by its deadline $d_p$, its weight $w_p$ is contributed to our objective. Our goal is to maximize weighted throughput in an online manner subject to the deadline constraints of packets and the varying fading channel qualities. We note here that our model can be an {\em overloaded system} --- it is possible that due to packets' deadline constraints, no algorithm can deliver all packets in the input instance, some packets have to be dropped.

There are two kinds of algorithms: {\em offline algorithms} and {\em online algorithms}. All input information (including the fade channel states, the packets' characteristics, and the packet sequence) is known to an offline algorithm in advance. Our work on offline algorithms in scheduling weighted packets with deadlines can be found in~\cite{LZ09a}. For an online algorithm, the packet input sequence is unknown. Each packet $p$'s characteristics (such as $r_p, \ w_p, \ d_p$) is known to the algorithm only at the time $r_p$ when $p$ actually arrives at the transmitter. Under various assumptions on the variants of the online version of this problem, the fade state of the channel is either completely unknown or partially known to the online algorithm. Note that essentially, delivering packets with deadlines in a wireless channel is an online decision-making problem. We address our online problem in a {\em preemption-restart} setting:
\begin{definition}
{\bf Preemption-restart setting.} An online algorithm is allowed to abort a packet during its transmission, and the aborted packet can be restarted (from scratch) and completed later in a preemption-restart setting.
\end{definition}


\section{Algorithms and Analysis}
\label{sec:online}

Scheduling packets with deadlines is essentially an online decision-making problem. In order to evaluate the worst-case performance of an online algorithm lacking of future input information, we compare it with an optimal offline algorithm. The offline algorithm is a clairvoyant algorithm, empowered to know the whole input sequence (including the fading states of the channel, the packet sequence, and every packet $p$'s characteristics $r_p, \ w_p, \ d_p$) in advance to make its decision. A competitive online algorithm, on the contrary, does not know the input sequence beforehand and it has a packet $p$'s characteristics only at the time $r_p$ when $p$ actually arrives. {\em In contrast to stochastic algorithms that provide statistical guarantees under some mild assumptions on input sequences, competitive online algorithms guarantee the worst-case performance.} Furthermore, when reasonable or reliable approximation of the input probability distribution is not available or when analytical worst-case performance guarantees are sought, competitive analysis is of fundamental research interests to us.

\begin{definition}
{\bf Competitive ratio}~\cite{BY98}. A deterministic online algorithm ${\tt ON}$ is called {\em $k$-competitive} if its weighted throughput on {\em any} finite instance $\cal I$ is at least $1 / k$ of the weighted throughput of an optimal offline algorithm on this instance:
\begin{displaymath}
k := \max_{\cal I} \frac{{\tt OPT}({\cal I}) - \delta}{{\tt ON}({\cal I})}
\end{displaymath}
where $\delta$ is a constant ($\delta$ becomes insignificant when the size of the input $|\cal I|$ increases) and ${\tt OPT}({\cal I})$ is the optimal offline solution of an input $\cal I$. The parameter $k$ is known as the online algorithm's {\em competitive ratio}. We also call the optimal offline algorithm {\em adversary}.
\end{definition}

The {\em upper bounds} of competitive ratios are achieved by some known online algorithms. A competitive ratio less than the {\em lower bound} is not reachable by any online algorithm. An online algorithm is said to be {\em optimal} if its competitive ratio reaches the lower bound. If the additive constant $\delta$ is no larger than $0$, the online algorithm ${\tt ON}$ is called {\em strictly $k$-competitive}. Competitiveness has been widely accepted as the metric to measure an online algorithm's worst-case performance in theoretical computer science and operations research~\cite{BY98}. In this section, we design and analyze competitive online scheduling algorithms in maximizing weighted throughput.

Without time constraints on packets, (weighted) throughput is maximized by simply delivering all packets that ever arrive at the transmitter. For packets with deadlines, when $q_t$ is a constant, an optimal competitive online algorithm can be achieved~\cite{CJST07a} such that the throughput (of uniform-value packets) is maximized. However, how to schedule packets with deadlines in fading channels still remains as an open problem and this problem becomes more interesting and complicated when packet weights are taken into account.


\subsection{The lower bound of competitive ratio}

We indicate if the fading states are unknown to the online algorithms, no online algorithm can have a competitive ratio better than $w_{\max} / w_{\min}$.

\begin{theorem}
\cite{LZ09a} If the fading states are unknown to online algorithms, no online algorithm can have a competitive ratio better than $w_{\max} / w_{\min}$.
\label{theorem:counter}
\end{theorem}

Based on Theorem~\ref{theorem:counter}, we know that if the fade states are completely unpredictable, without one step of look-ahead, no online algorithm can have a competitive ratio better than $w_{\max} / w_{\min}$. In the following, we consider a practical scenario and make the following assumption that is widely accepted:
\begin{assumption}
\cite{TH98, TSZM07, ZM08} The online algorithms have the ability of looking one-step ahead of knowing the fade states of the channel. At the time when an online algorithm starts to schedule a packet from the current time, this ``committed'' packet can be sent successfully according to future fading states. However, the online algorithm is allowed to preempt-restart this packet later and this packet is not guaranteed to be sent eventually if it is preempted.
\label{assumption:1}
\end{assumption}

Assumption applies to all the variants we consider in the following.


\subsection{Assume both the fade states and the packet input sequence are unknown to the online algorithms}

Two packets are with more interests when neither the fade states nor the input sequence is known to the online algorithm:
\begin{enumerate}
\item $i$: the currently running packet. If $i$ is not available, we simply create a virtual packet $i$ with $w_i = 0$.
\item $h$: the packet with the maximum-value among all pending packets for the transmitter.
\end{enumerate}

From~\cite{LZ09a}, we conclude that always sending the packet with the earliest deadline or $h$ results a competitive ratio arbitrarily large. Here, we employ the following ideas of getting an online algorithm with a better competitive ratio: If the currently sending packet is with a sufficiently large value, then we keep sending it. Otherwise, we let $h$ preempt it. The algorithm we study is called {\tt SEMI-GREEDY}.

\begin{algorithm}
\caption{{\tt SEMI-GREEDY}($\alpha > 1$)}
\begin{algorithmic}[1]

\STATE Let the maximum-value pending packet be $h$, with ties broken in favor of the earliest deadlines. Let the currently being sent packet be $i$. If $h$ (or $i$) does not exist, we set $w_h = 0$ (or $w_i = 0$).

\IF{$w_h \ge \alpha \cdot w_i$}

\STATE Abort $i$ and send $h$.

\ENDIF

\end{algorithmic}
\label{alg:on}
\end{algorithm}

Before we prove the competitive ratio for the algorithm {\tt SEMI-GREEDY}, we define a concept that is useful to the proof.

\begin{definition}
{\bf Packet chains}. We define a packet chain $C$ of $k$ packets as $C := \{p_1, \ p_2, \ p_3, \ \ldots, \ p_k\}$ with the following property ($\alpha > 1$), $w_{p_i} \ \le \ w_{p_{i + 1}} \ / \ \alpha, \ \ \ \forall \ i \ = \ 1, \ 2, \ 3, \ \ldots, \ k - 1$. We use $W(C)$ to represent the total value of the packets of $C$.
\end{definition}

\begin{lemma}
Given a chain $C$ of $k \ge 2$ packets $p_1, \ p_2, \ \ldots, p_k$, we have $W(C) \ \le \ (\frac{1}{\alpha - 1} \cdot (\alpha^{n + 1} - 1) / \alpha^n) \cdot w_{p_k}$.
\label{lemma:chain}
\end{lemma}

\begin{proof}
\begin{eqnarray*}
\frac{W(C)}{w_{p_k}} & = & \frac{\sum^k_{i = 1} w_{p_i}}{w_{p_k}} = \frac{w_{p_1} + w_{p_2} + \cdots + w_{p_{k - 1}} + w_{p_k}}{w_{p_k}} \\
& \le & \frac{w_{p_1} + w_{p_2} + \cdots + w_{p_{k - 1}} + \alpha \cdot w_{p_{k - 1}}}{\alpha \cdot w_{p_{k - 1}}} \\
& = & 1 + \frac{1}{\alpha} \cdot \frac{w_{p_1} + w_{p_2} + \cdots + w_{p_{k - 1}}}{w_{p_{k - 1}}} \\
& \le & \ldots \\
& \le & 1 + \frac{1}{\alpha} + \frac{1}{\alpha^2} + \cdots + \frac{1}{\alpha^{k - 2}} + \frac{1}{\alpha^{k - 1}} + \frac{1}{\alpha^{k}} \\
& = & \frac{1}{\alpha - 1} \cdot (\alpha^{k + 1} - 1) / \alpha^k
\end{eqnarray*}
\end{proof}

\begin{theorem}
The {\tt SEMI-GREEDY} algorithm has a competitive ratio of $\max\{1 + \alpha, \ \frac{1}{\alpha - 1} \cdot (\alpha^{n + 1} - 1) / \alpha^n)\}$. It is ($\phi^2 \approx 2.618$)-competitive when $\alpha = \phi \approx 1.618$.
\label{theorem:general}
\end{theorem}

\begin{proof}
We use a charging scheme to prove Theorem~\ref{theorem:general}. The idea is: For the packets the adversary sends, we charge them into different time intervals and we prove that in each pair of corresponding intervals, the value we charge to the adversary in that interval is no more than $\max\{1 + \alpha, \ \frac{1}{\alpha - 1} \cdot (\alpha^{n + 1} - 1) / \alpha^n)\}$ times of what {\tt SEMI-GREEDY} achieves.

Let the subset of packets chosen by the adversary (that is, an optimal offline algorithm) (respectively, {\tt SEMI-GREEDY}) be $\Pi_1$ (respectively, $\Pi_2$). Without loss of generality, we assume the adversary sends packets in a {\em canonical order}, i.e., for any two pending packets $p_i$ and $p_j$, the adversary sends the packet with an earlier deadline. We are going to prove that
\begin{displaymath}
\frac{\sum_{p_j \in \Pi_1} w_{p_j}}{\sum_{p_i \in \Pi_2} w_{p_i}} \le \max\{1 + \alpha, \ \frac{1}{\alpha - 1} \cdot (\alpha^{n + 1} - 1) / \alpha^n)\}.
\end{displaymath}

The proof depends on the following two observations:
\begin{enumerate}
\item Given a set of packets $S$ at time $t$, we assume an online algorithm schedules a packet $p_i \in S$. We consider time $t' > t$. Since all packets are with the same length, if the packet $p_i$ cannot be finished by time $t'$, any packet in $S$ cannot be finished completely by time $t'$, no matter what the fade states of the channel has.

\item Given a set of packets $S$ at time $t$, we assume that the {\tt SEMI-GREEDY} algorithm schedules a packet $p_i \in S$. We have $w_{p_i} \ge \max_{p_j \in S} w_{p_j} / \alpha$.

    If we assume $p_i$ is aborted at time $t' > t$ by a packet $p_k$, we have $w_{p_i} < w_{p_k} / \alpha$ and $p_k \notin S$. If the preempting packet $p_k$ is not sent by the algorithm {\tt SEMI-GREEDY}, $p_k$ must be aborted by another packet which has the potential of being sent. So on and so forth, we regard all aborted packets and the last-sent packet $p_l$ as a chain. From Lemma~\ref{lemma:chain}, all ever-aborted packets have value $\le w_{p_l} \cdot \frac{1}{\alpha - 1} \cdot (\alpha^{n + 1} - 1) / \alpha^n$.

    Note that no chains share a same packet, since each preempted packet and its preempting packet are in the same chain.
\end{enumerate}

For any packet $p \in (\Pi_1 \setminus \Pi_2)$ sent only by the optimal offline algorithm, either $p$ expires before {\tt SEMI-GREEDY} sends it or $p$ was sent, {\tt SEMI-GREEDY} aborted $p$ before $p$ could be finished, and $p$ is never completed by its deadline. If $p$ expires, any packet that {\tt SEMI-GREEDY} sends since time $r_p$ has a value $\ge w_p / \alpha$ (from the algorithm). For each time interval in which a single packet is sent, we examine it for both the optimal offline algorithm and this online {\tt SEMI-GREEDY} algorithm in a sequential order. Our charging scheme works as follows:
\begin{enumerate}
\item For any packet $p \in (\Pi_1 \setminus \Pi_2)$ that {\tt SEMI-GREEDY} has not ever run, we charge it to the corresponding time interval that {\tt SEMI-GREEDY} sends a packet. We note that {\tt SEMI-GREEDY} must have one pending packet to send in this interval since this packet $p$ is a candidate. The packet {\tt SEMI-GREEDY} sends, let it be $p'$, in this corresponding interval has a value no less than $w_p / \alpha$. Also, {\tt SEMI-GREEDY} finishes $p'$ no later than the adversary finishes $p$ since $p$ and $p'$ have the same processing time and $p$ and $p'$ are being executed in corresponding time intervals when both algorithms send packets.

\item For any packet $p \in (\Pi_1 \setminus \Pi_2)$ that {\tt SEMI-GREEDY} ever sends but aborts it later, we know that (from above observations) $p$ belongs uniquely to a chain and the last element of this chain, say $p'$, is sent by {\tt SEMI-GREEDY}. Thus, we charge $w_p$ to the time interval when $p'$ is sent by {\tt SEMI-GREEDY}.

\item For any packet $p \in (\Pi_1 \cap \Pi_2)$, we charge $w_p$ to the time interval when {\tt SEMI-GREEDY} sends $p$. Clearly, for any packet acting as the last-element of a chain, this charging scheme results that the value ratio is bounded by $\frac{1}{\alpha - 1} \cdot (\alpha^{n + 1} - 1) / \alpha^n$ (see Lemma~\ref{lemma:chain}).
\end{enumerate}

The remaining part of the proof is to argue that when we charge a packet $p \in (\Pi_1 \setminus \Pi_2)$ that {\tt SEMI-GREEDY} has not ever run yet in the corresponding time interval, {\tt SEMI-GREEDY} sends a packet $p'$, $w_{p'} \ge w_p / \alpha$. This claim is easy to prove since if $w_{p'} < w_p / \alpha$, $p'$ will be aborted by $p$ immediately at the time when $p$ arrives.  Thus, for each packet $p$ that {\tt SEMI-GREEDY} sends, the charged value to $p$ for the adversary is bounded by $1 + \alpha$ and $\frac{1}{\alpha - 1} \cdot (\alpha^{n + 1} - 1) / \alpha^n$ and all packet that the adversary sends have been charged. Theorem~\ref{theorem:general} is proved.
\end{proof}

Closing or shrinking the gap $[2, \ 2.618]$ for deterministic online algorithms still remains as an open problem.


\subsection{Assume the fade states are known to the online algorithms, but the packet input sequence is unknown}

Now we consider a variant in which the fade states are known beforehand, but the packet input sequence is unknown. To illustrate the challenge, we present an instance in which packets are with the same value and the fade state of the channel is fixed at $q_t = l / 2, \forall t$. Consider one packet $p_1$ with deadline $5$ at time $1$. If an online algorithm executes $p_1$, the adversary releases another packet $p_2$ with deadline $3$ at time $2$. So, the online algorithm cannot finish both jobs and the competitive ratio is $2$, given the adversary finishes both packets in the order of $p_2$ and $p_1$. If the online algorithm aborts $p_1$ but executes $p_2$ at time $2$, the adversary releases another packet $p_3$ at time $2$ with deadline $4$. Here, the online algorithm cannot finish both $p_2$ and $p_3$, but the adversary can finish $p_1$ and $p_3$ by their deadlines in order. Thus, the lower bound of competitive ratios for this variant ($w_{p_i} = 1$, $\forall i$ and fade states keep constant $1 / 2$) is $2$. It is intuitive to abort a running packet if it can be sent later with respect to the given set of pending packets and fade states of the channel. Our proposed online algorithms are based on this intuition. In order to check if a set of packets can be delivered successfully, we define
\begin{definition}
{\bf Provisional schedule}~\cite{CJST07a, EW07}. At any time $t$, a {\em provisional schedule} ${\bf S}_t$ is a schedule for the pending packets at time $t$ (assuming no new arriving packets). This schedule specifies the set of packets to be transmitted, and for each it specifies the delivery time. An optimal provisional schedule is the one achieving the maximum total value of packets among all provisional schedules.
\end{definition}

In the following, we provide a modified earliest-deadline-first algorithm called $\tt EDF_{\beta}$. Since the fade states are known, there exists an efficient algorithm in calculating an optimal provisional schedule for time $t$ (see~\cite{LZ09a}). We are interested in two packets in this provisional schedule: the earliest-deadline pending packet $e$ and the packet $h$ with the maximum value. We either schedule $e$ (if $e$ is with a sufficiently large value) or another packet $f$ satisfying $w_f \ge \max\{\beta \cdot w_e, \ w_h / \beta\}$.

\begin{algorithm}
\caption{$\tt EDF_{\beta}$}
\begin{algorithmic}[1]

\STATE Abort the currently running packet $i$ only if the new arrival with value $\ge \beta \cdot w_i$, ties are broken in favor of the packet with the earliest deadline.

\IF{there is no currently running packet}

\STATE Calculate the optimal provisional schedule, based on the set of pending packets and the known fade states.

\IF{$w_e \ge w_h / \beta$}

\STATE Execute $e$.

\ELSE

\STATE Execute a packet $f$ satisfying
\begin{displaymath}
w_f \ge \max\{\beta \cdot w_e, \ w_h / \beta\}.
\end{displaymath}
where ties are broken in favor of the earliest-deadline packet. Note $h$ itself is a candidate for $f$.

\ENDIF

\ENDIF

\end{algorithmic}
\label{alg:random}
\end{algorithm}

\begin{theorem}
Assume fade states are known to online algorithms. Algorithm ${\tt EDF}_{\beta}$ is $\max\{2, \ \beta, \ (\frac{1}{\beta - 1} \cdot (\beta^{n + 1} - 1) / \beta^n)\}$-competitive, and it is $2$-competitive when $\beta = 2$.
\label{theorem:bare}
\end{theorem}

\begin{proof}
We use a potential function method and a loop invariant method to prove Theorem~\ref{theorem:bare}. We compare our algorithm ${\tt EDF}_{\beta}$ with the adversary {\tt ADV}. Let $\Phi^{\tt ADV}_t$ and $\Phi^{\tt EDF}_t$ denote the potentials of the adversary and ${\tt EDF}_{\beta}$ respectively. Specifically, $\Phi^{\tt ADV}_t$ denotes the total value achieved since time $t$ from the pending packets at time $t$ for the adversary. Let this set of packets be $S^*_t$. Let $\Phi^{\tt EDF}_t$ denote the total value of the optimal petitional schedule of the pending packets at time $t$ for ${\tt EDF}_{\beta}$. We use $p_t$ and $p'_t$ to denote the $t$-th packet sent by ${\tt EDF}_{\beta}$ and {\tt ADV} respectively. If such a packet does not exist, $p_t$ ($p'_t$) is a null packet with value $0$. To prove Theorem~\ref{theorem:bare}, we need to show that for any $t$, $c \cdot w_{p_t} + \Delta \Phi^{\tt EDF}_t \ge w_{p'_t} + \Delta \Phi^{\tt ADV}_t$, where $c := \max\{2, \ \beta, \ (\frac{1}{\beta - 1} \cdot (\beta^{n + 1} - 1) / \beta^n)\}$. We provide the following loop invariants and prove their correctness by case study.

\begin{itemize}
\item Denote the pending packets at time $t$ for {\tt ADV} and ${\tt EDF}_{\beta}$ as ${\cal P}'_t$ and ${\cal P}_t$. ${\cal P}'_t \subseteq {\cal P}_t$. Note that ${\tt EDF}_{\beta}$ may not deliver all packets in ${\cal P}_t$.

\item When a packet is sent, the sum of the actual gain and the credit charge (see below) is called {\em amortized gain}. We prove that for the $i$-th packet sent, {\tt ADV}'s amortized gain is no more than $c$ times of ${\tt EDF}_{\beta}$'s amortized gain.
\end{itemize}

For arrivals, with the first invariant, the invariants are easy to prove. Note $w_{p_t} = w_{p'_t} = 0$. In the following, we consider packet deliveries only. Let the packet ${\tt EDF}_{\beta}$ chooses to send in this time interval be $p$. One fact that we will use in the proof is: Given two packet $p$ and a packet $p^*$ with $d_p \le d_{p^*}$, if $p$ is not in the optimal provisional schedule, but $p^*$ is, then $w_{p^*} \ge w_p$. This fact further implies that if $p$ is the packet ${\tt EDF}_{\beta}$ is currently sending, any packet not in the optimal provisional schedule has a value $\le \beta \cdot w_p$.

\begin{enumerate}

\item Assume {\tt ADV} sends a packet $p'$. Assume $p$ is sent successfully by ${\tt EDF}_{\beta}$.

Based on the invariants, $w_{p'}, \ w_p \ \le \ w_h$. From the algorithm itself, $w_p \ge w_h / \beta$. Since all packets have the same length, under any fade states, ${\tt EDF}_{\beta}$ finishes $p$ no later than {\tt ADV} finishes $p'$.

If $d_{p'} < d_p$, we have $w_{p'} < w_p$ in the optimal provisional schedule. Then we charge $w_{p'} + w_p$ to the adversary and we have $w_{p'} + w_p \le 2 \cdot w_p$. If $d_{p'} > d_p$, $p$ will not be sent by the adversary. Then we charge $w_{p'}$ to {\tt ADV} and we have $\beta \cdot w_p \ge w_h \ge w_{p'}$.

\item Assume {\tt ADV} sends a packet $p'$. Assume $p$ is aborted by ${\tt EDF}_{\beta}$ before it is finished.

If the adversary will send $p$, we will charge $w_p$ to the packet that preempts it. Like the chain we have calculated in Lemma~\ref{lemma:chain}, the value gained by sending the last packet of the chain is at least $(\beta - 1) \cdot \beta^n / (\beta^{n + 1} - 1)$ times of the total value we charge for the adversary.

\item Assume {\tt ADV} has nothing to send from the currently pending packets for ${\tt EDF}_{\beta}$.

We claim that either $p$ has been sent by {\tt ADV} or {\tt ADV} must have one new arrival before ${\tt EDF}_{\beta}$ finishes the packet $p$ it chooses to send. Otherwise, {\tt ADV} can get more credit by delivering $p$. It does not hurt if we have run $p$ till new arrivals come. This analysis is similar to what we have had in above cases.
\end{enumerate}

\end{proof}

Theorem~\ref{theorem:bare} implies that extra information (for example, the known fade states) helps improve competitive ratio.


\subsection{Assume the fade states are unknown, but the packet input sequence is known}

We first provide the lower bound $\phi \approx 1.618$ of competitive ratio for deterministic online algorithms for this variant. Then we discuss the relation between this model and another well-studied online problem.

\begin{theorem}
Consider a variant in which the fade states are unknown, but the packet input sequence is known to online algorithms. The lower bound of competitive ratio for deterministic online algorithms is $\phi \approx 1.618$.
\end{theorem}

\begin{proof}
Such an instance is easy to construct. Assume there are two packets in the input sequence only. One packet $p_1$ is with value $1$ and deadline $2$. The other packet $p_2$ is with value $\phi$ and deadline $3$. These two packets are released at time $0$. Let an online algorithm be {\tt ON}.

If {\tt ON} schedules $p_1$, the optimal offline algorithm schedules $p_2$ and the fade states are $0.5$ from time $0$ to $3$. Note here the Assumption~\ref{assumption:1} still holds. Then the competitive ratio is $\phi$. If {\tt ON} schedules $p_2$, the optimal offline algorithm schedules both $p_1$ and $p_2$ given the fading states are $0.5$ from $0$ to $4$. Thus, the competitive ratio is $(1 + \phi) / \phi = \phi$.
\end{proof}

In the following, we reveal the relationship between this variant and a well-studied model called the {\em bounded-delay model} (see~\cite{EW07, LSS07} and the references therein). In the bounded-delay model, packets are released in an online manner. Each packet is associated with a value and a deadline by which it should be sent. In each time step, a packet can be sent and the goal is to maximize the total value of the packets sent by their respective deadlines.

\begin{theorem}
Assume the fade states are unknown, but the packet input sequence is known to online algorithms. A $c$-competitive algorithm for the bounded-delay model implies a $c$-competitive algorithm for our model.
\label{theorem:3}
\end{theorem}

\begin{proof}
Consider an input sequence $\cal I$ for the bounded-delay model. Let the packets sent by an optimal offline algorithm be $\cal O$ and the algorithm be ${\tt OPT}_d$.

Given a time $t$, we create the fade states such that the optimal offline algorithm ${\tt OPT}_f$ for this variant achieves the same weighted throughput as ${\tt OPT}_d$, also, for an online algorithm, the extra given information about the whole input sequence cannot avoid the difficulty brought by the unpredictable fade states. The construction of fade states is as follows.

For the bounded-delay model, let the set of packets ${\cal O}$ be $p_1, \ p_2, \ \ldots, \ p_m$ and they are sent in time steps $1, \ 2, \ \ldots, \ m$ respectively. (If there is no packet sent in a step $i$, we create a dummy packet $p_i$ for step $i$ with $w_{p_i} = 0$.) Without loss of generality, all packets $p_i$ can be sent in the earliest-deadline-first manner. Then we modify the deadlines of the packets in $\cal O$ such that $d_{p_i} < \min\{d_{p_{i + 1}}, \ \ldots, \ d_{p_m}\}$, for all $i = 1, 2, \ldots, m - 1$. At last, we force the quality of the fade states from time $d_{p_i}$ to $d_{p_{i + 1}}$ be $l / (d_{p_{i + 1}} - d_{p_i})$. This guarantees a packet can be sent under such fade states and if $p_i$ is pending to an online algorithm at time $d_{p_{i - 1}}$ and the online algorithm sends any other packet than $p_i$, $p_i$ cannot be sent by the online algorithm any more. We ensure that the optimal offline algorithm for this variant works the same as the optimal offline algorithm for the bounded-delay model. Also, the extra information about the packet input sequence does not help the online algorithm since it has no knowledge about the fade states. With Assumption~\ref{assumption:1}, the online algorithm known that only one packet can be sent once it is committed and this is exactly as what is assumed in the bounded-delay model.
\end{proof}

Closing or shrink the gap of competitive ratios $[1.618, \ 1.832]$ for the bounded-delay model is an intriguing problem and thus, from Theorem~\ref{theorem:3}, the gap still applies to the variant in which the fade states are unknown, but the packet input sequence is known to online algorithms.


\bibliographystyle{plain}
\bibliography{buffer}


\end{document}